\documentclass[a4paper,UKenglish,cleveref, autoref, thm-restate,authorcolumns]{lipics-v2019}

\nolinenumbers

\title{Safe Subjoins in Acyclic Joins} 

\titlerunning{Safe Subjoins in Acyclic Joins} 

\author{Foto N. Afrati}{National Technical University of Athens,  Greece}{afrati@gmail.com}{[orcid]}{}

\authorrunning{Foto N. Afrati} 

\Copyright{Foto N. Afrati} 



\category{} 

\relatedversion{} 

\supplement{}

\hideLIPIcs  

\begin{document}

\maketitle

\begin{abstract}
It is expensive to compute joins, often due to large intermediate relations. For acyclic joins, monotone join expressions are guaranteed to produce intermediate relations not larger than the size of the output of the join  when it is computed on a fully reduced database. 
Any subexpression of an acyclic join does not offer this guarantee, as it is easy to prove. 
In this paper, we consider joins with projections too and we ask the question whether we can characterize join subexpressions that produce, on every fully reduced database, an output  without dangling tuples (which translates, in the case of joins without projections, to an output of size not larger than the size of the output of the join). We call such a subexpression a safe subjoin. Surprisingly, we  prove that there is a simple characterization which is the following: A subjoin is safe if and only if there is a parse tree of the join (a.k.a. join tree)  such that the relations in the subjoin form a subtree of it. We provide an algorithm that finds  such a parse tree, if there is one.
\end{abstract}

\section{Introduction}

Computing a join efficiently is one of the fundamental problems in database systems.
Acyclic joins \cite{Fagin83}  have been extensively investigated and their properties  enable optimizations in classical well studied problems but also in various modern contexts, such as machine  learning  (\cite{SchleichOK0N19}, \cite{SchleichOC16}).
Relatively recent work includes the development of I/O optimal algorithms for acyclic joins \cite{XiaKe16} \cite{PaghP06}.  These works  assume that the relations are fully reduced and this is our assumption too here.

In many cases, when computing joins, it is critical to study and decide 
the join ordering problem
(e.g., see \cite{MoerkotteN06}, \cite{Trummer017}).
When we have an acyclic join, we know that with a certain polynomial time preprocessing which derives a fully reduced database instance, there is a certain order of computing the join that guarantees sizes of intermediate relations to be smaller than the size of the output of the join. However, the optimal order of computing an acyclic join is not known. E.g., when can we push larger relations to join in the end of the join process, without compromising the property that sizes of intermediate relations are smaller than the size of the output of the join? This depends on properties of subjoins of an acyclic join. In that respect we study here the following problem:

\begin{itemize}
\item 
 When a subjoin of an acyclic join is guaranteed not to compute dangling tuples over a fully reduced database instance?
\end{itemize}

A dangling tuple is a tuple of a relation or of a subjoin which is not used in the computation of the join, i.e., if deleted, the output of the join will be the same.
Interestingly we give a complete characterization of such subjoins. 
We illustrate the problem on an example:

\begin{example}
\label{ex-begin}

We consider the join $J=ABC\bowtie AB \bowtie AC \bowtie BC$. This is an acyclic join.
We consider  subjoin, $J_S= AB \bowtie AC \bowtie BC$,  that includes only the last three relations.
This subjoin has an undesirable property. We will explain on the following database instance $D$:\\
The relation $r_0=ABC$  has the tuples $\{  (a,b,c_1), (a,b_1,c), (a_1,b,c)   \}$.\\
The relation $r_1=AB$ has the tuples $\{ (a,b), (a,b_1), (a_1,b) \}$.\\
The relation $r_2=AC$ has the tuples $\{ (a,c),(a,c_1), (a_1,c) \}$.\\
The relation $r_3=BC$ has the tuples $\{ (b,c),(b,c_1), (b_1,c) \}$.

Database $D$ is fully reduced, i.e., there are no dangling tuples in $D$.

Now it is easy to observe that the output of the $J_S= AB \bowtie AC \bowtie BC$ contains 4 tuples, while the output of the whole join contains 3 tuples. The tuple $(a,b,c)$ computed in the output of the subjoin is a dangling tuple, i.e, it is not used in the computation of the  join.
\end{example}

A byproduct of the techniques developed in this paper is presented in Section~\ref{sec-min}. It is 
work towards characterizing subjoins that contain the minimum number of subsubjoins that can be processed without  each of it producing dangling tuples.


\subsection{Problem Definition}

Let $J$ be an acyclic join and let $J_S$ be a join, called a {\em subjoin} of $J$ here on, which results from $J$ after deleting some relations with properly chosen attributes to be projected to appear in the output as follows: these are a) the attributes that are projected in the output of $J$ and belong to some relation in the subjoin  and b) the boundary attributes. The {\em boundary attributes} are the ones that belong to both a) relations of the subjoin and to b) relations  of $J$ that are not in the subjoin. We conveniently define the complement $J^c_S$, of the subjoin $J_S$  to be the subjoin of $J$ which uses the relations that are not in $J_S$. An example is presented in 
Appendix \ref{sec-projection}.

We say that a relation $r$ has no dangling tuples with respect to a relation $r'$ if every tuple in $r$ joins with a tuple in $r'$ to produce a tuple in the output of $r\bowtie r'$.

We call the subjoin $J_S$
{\em safe} if the following is true. For every fully reduced (i.e.,  consistent) database D, the output $J_S(D)$ of $J_S$ computed on D has the property that $J_S(D)$ has no dangling tuples with respect to $J^c_S(D)$.

When the join has no projections (i.e., all attributes appear in the output), then  the following is also true for a safe subjoin: Every tuple
$t_S$ in $J_S(D)$ is such that there is a tuple $t$ in $J(D)$ such that $t_S = t[A_S]$ where
$t[A_S]$ is the projection of $t$ on the attributes in $A_S$, where   $A_S$ is  the set of attributes that appear in $J_S$.
When the set of attributes $A_S$ is evident from the context we use the term {\em subtuple} of $t$ to refer to $t[A_S]$.

The problem was introduced by Christopher R\'e  \cite{Resafe}.  Example~\ref{ex-acyclic-non-safe}, in Appendix, shows that a nonsafe subjoin can be acyclic. In the following subsection we break down the proof.
\vspace*{-.1cm}
\subsection{Components of the proof}
\label{subsec-components}
The structure of an acyclic join is given by a parse tree (a.k.a. join tree). 
We use parse trees to characterize safe subjoins.
We prove that  a subjoin is safe if and only if there is a parse tree of the join such that the relations in the subjoin form a partial subtree of it. 

The proof procedure considers an arbitrary parse tree of the join and either transforms it into a parse where the subjoin forms a single 
partial subtree or it builds a counterexample database to prove that the subjoin is not safe. 
More specifically, 
given an acyclic join $J$ and a subjoin $J_S$  we consider two cases depending on whether the following is true or not:

{\bf Property:} There is a relation $r$ that does not belong to the subjoin such that there is no relation $r'$ that belongs to the subjoin for which the following happens: $r\cap r'$ contains all the attributes of $r$ that appear in at least one of the relations of the subjoin.

Thus the two main blocks of the proof are the following:
\begin{itemize}
\item
Case 1. When  the above property is true. Then $J_S$ is not safe.
This  is proven in Section  \ref{sec-5}; a  counterexample database is built using tuple generating dependencies and chase. 

\item
Case 2. When  the above property is false. 

Then, given a parse tree $T$ of $J$ with at least two disconnected parts of $J_S$, either there is parse tree of $J_S$ where the number of disconnected parts is less that the one in $T$ or  $J_S$  is not safe.
This is proven in Sections  \ref{sec-counter-shared-sec} and \ref{sec-warmup}  using the result of
Subsection \ref{subsec-reverse}. 

Most of the insight of the proof of this second case   can be obtained by considering the simplest subjoin that is  partitioned into two disconnected parts  and this is what is presented in Section  \ref{sec-warmup}.  The disconnected parts mentioned above are defined formally as maximal subtrees in Subsection \ref{sec-23}.
\end{itemize}

A result of independent interest is the reverse path transformation in  Subsection \ref{subsec-reverse} which transforms a parse tree of the acyclic join hypergraph to another parse tree of it. 

%
%
%
%
%
%

\section{Preliminary Definitions and Technical Tools}

\subsection{Preliminaries}

This subsection contains definitions and results from the literature. For more details see, e.g., \cite{abiteboulbook,GottlobLS01,Ullman01-1,Ullman01-12,JEFW}.


We define a hypergraph $G$  as a pair $ (V,E)$  where:
$ V$  is a finite set of vertices and 
$ E$  is a set of hyperedges, each hyperedge being a nonempty subset of $ V$ .
We will refer to hyperedges as edges henceforth.
A {\em hypergraph of a join} has vertices that correspond to its attributes and there is an hyperedge (hereon, referred to as, simply, edge) joining a subset of the attributes (hereon, we will refer to the vertices of a hypergraph as attributes) if there is a relation in the join which contains exactly this subset of attributes. 
We compute a join $J$ on a database $D$ by assigning values to the attributes of $J$ such that the tuples that are obtained by this assignment belong to the corresponding relations in $D$.

\begin{definition}
\label{dfn-acyclicc}
A join is {\em acyclic} if there is a tree with nodes representing the edges (the relations, respectively) of the hypergraph (the join, respectively) where the following is true:
 For each attributes $A$, all the nodes of the tree where $A$ appears are connected.  We call such a tree a {\em parse tree} 
(or join tree). 
\end{definition}


There is a lot of  early and recent work on acyclic joins, e.g.,
\cite{YuLSO79,Graham79,TarjanY84,WangZA000Z21,WangY20}.
An example of a parse tree is in the Appendix \ref{app-ex}.   
On a parse tree, the {\em depth} of a node is its distance from the root of the parse tree. Also, when we refer to a {\em subtree rooted} at a certain node $u$ we mean the subtree that is equal to the set of all descendants of $u$ in the parse tree. 
In the rest of the paper we will refer to relations of a join, edges of its hypergraph and nodes of a parse tree of the join interchangably, thus a node of a parse tree represents also a set of attributes. 
%
%
%
%
%

\begin{definition}
A database instance, $D$, is {\em  consistent for $J$} or simply {\em consistent}  (if $J$ is obvious) if every relation instance in $D$
is the projection of the output of $J$ applied on $D$.

$D$ is {\em pairwise consistent  for $J$} or simply {\em pairwise consistent}  if every pair of relations, $r_i, r_j$, in $J$ that share at least one attribute are
consistent, i.e., each  relation instance in the pair
is the projection of $r_i \bowtie r_j$ applied on $D$.
\end{definition}

%

In Section \ref{sec-semijoin} in the Appendix, we include a short presentation of the role of semijoins in producing a fully reduced database.

\begin{definition}
A {\em path} from a vertex $u$ to a vertex $v$ is a sequence of $k$ edges $E_1, \ldots, E_k$ such that
$u$ is in $E_1$ and $v$ is in $E_k$ and for each $i=1,\ldots, k-1$, the intersection of
$E_i$ with $E_{i+1}$ is nonempty. We also say that the above sequence is a path from edge $E_1$ to edge $E_k$.
\end{definition}

\begin{definition}
Two  {\em vertices are connected} if there is a path from one to the other. Similarly, two  {\em edges are connected} if there is a path from one to the other. A  {\em set of vertices (or a set of edges) is connected} if there is a path joining every pair of vertices (or edges) in the set.

The  {\em connected components} of a hypergraph are the maximal connected sets of edges.
\end{definition}



\begin{definition}
Let $N_1$ be a subset of the vertices of a hypergraph. The set of {\em partial edges generated by $N_1$} is the set of edges obtained by intersecting each edge with $N_1$. 
\end{definition}

\subsection{Technical tool: Reverse Path Transformation}
\label{subsec-reverse}

Our first contribution is 
Lemma \ref{stem-lemma} which  is one of the main tools and is of independent interest. It provides the necessary condition for a certain transformation on a parse tree of an acyclic join.  

Let $T$ be a parse tree of an acyclic join hypergraph. 
We say that a path $p$ satisfies the {\em shared-attributes condition} if $p=a_1, a_2, \ldots ,a_n$ where  $a_i$ is the parent of $a_{i+1}$ in T, $i=1,2,\ldots , n-1$ and 
$P_1\cap a_1=P_1\cap a_2 = \cdots = P_1\cap a_n$,
where $P_1$ is the parent of $a_1$.

\begin{lemma}
\label{stem-lemma}

Let $T$ be a parse tree of
an acyclic join $J$. Suppose path $p=a_1, a_2, \ldots ,a_n$ satisfies the shared-attributes condition.

Let $T_n$ be the subtree rooted at $a_n $.
Let $T_i, i=1, 2, \ldots, n-1$, be the subtree rooted
at $a_i$ after removing its child $a_{i+1}$ together with
the subtree rooted at $a_{i+1}$.

Then, there is another parse tree $T'$ of join $J$ such that the parent of $a_n$ is $P_1$
the parent of $a_{n-1}$ is $a_n$, the parent of $a_i $ is $a_{i-1}$ and, the sub-tree $T_i $ is rooted at $a_i$, $i=1,\ldots,n$.
See Figure \ref{fig:diagram:wps} for an illustration.
\end{lemma}


\begin{figure}

\begin{center}\begin{tabular}{l l}

\includegraphics[height=7cm, angle=0, keepaspectratio=true,clip]{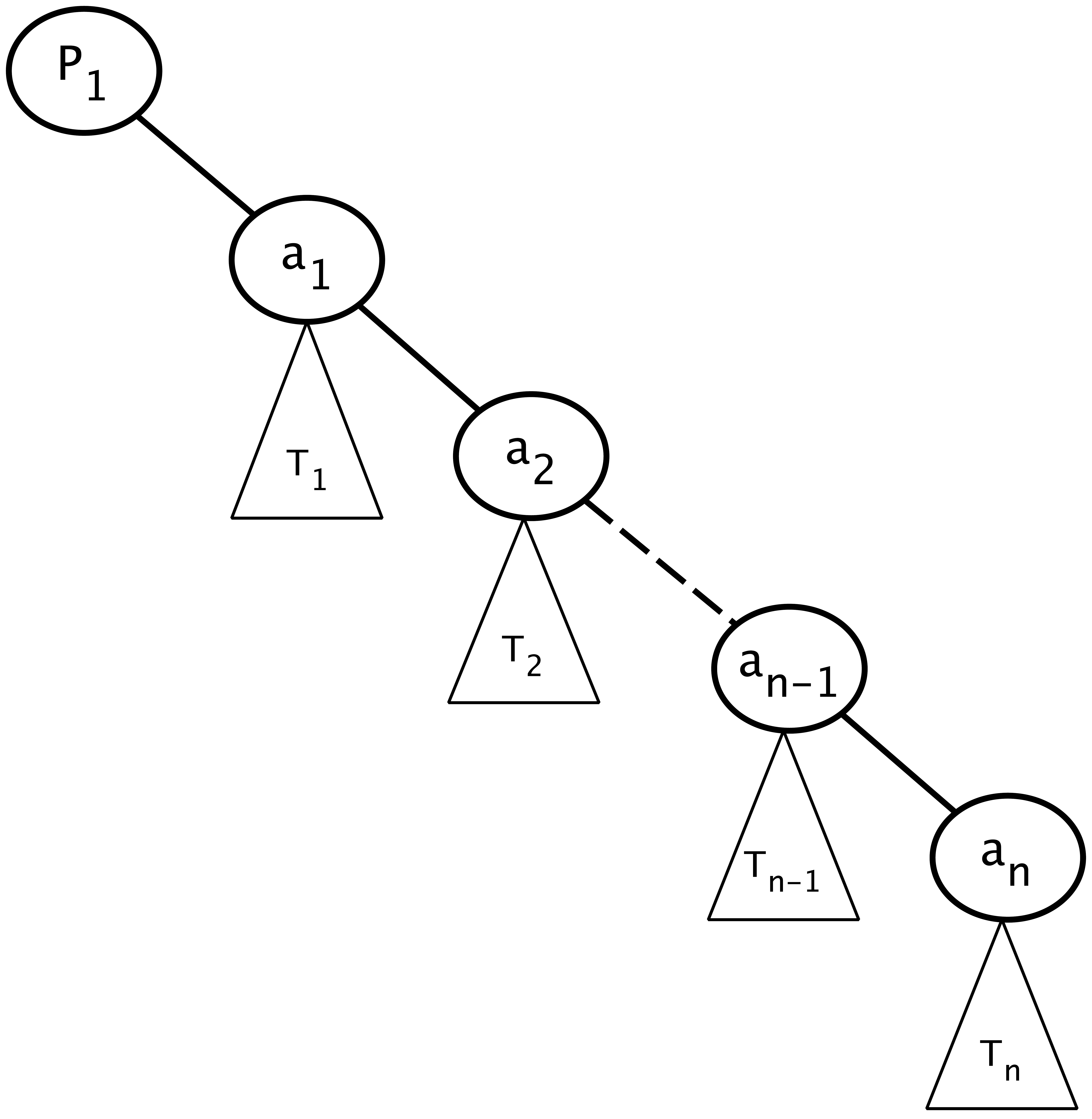}


%
%
%
%


&

\includegraphics[height=7cm, angle=0, keepaspectratio=true,clip]{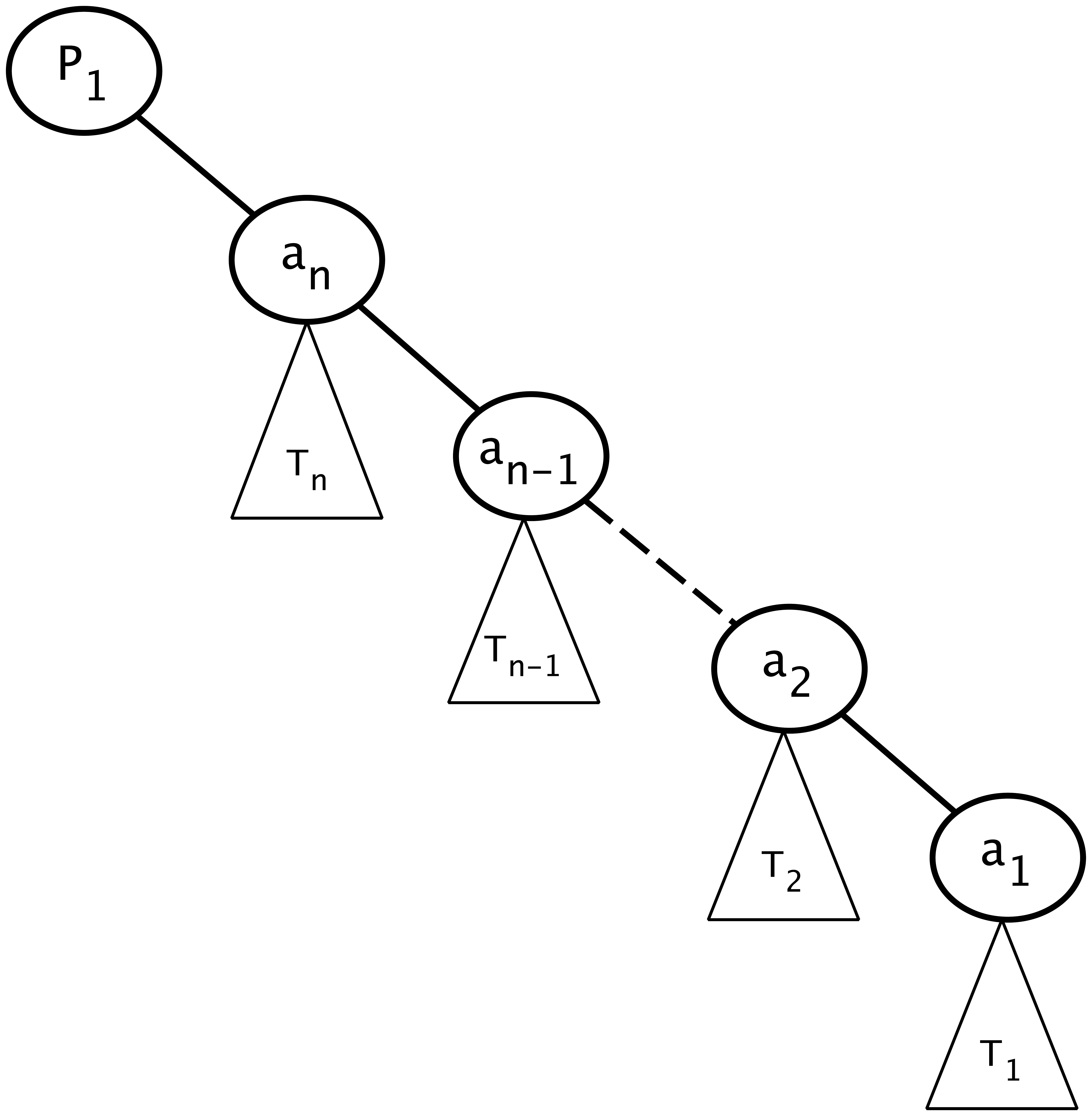}

\\
~~~~~~~~~~~~(a) & ~~~~~~~~~~~~~~~~~~~~~~~~~~ (b) \\
\end{tabular}\end{center}
\caption{(a) is the parse tree we  start with and (b) is the transformed parse tree. I.e., we delete $(P_1,a_1)$ and we add $(P_1,a_n)$.\label{fig:diagram:wps}} 
\end{figure}

\begin{proof}
We delete $(P_1,a_1)$ and we add $(P_1,a_n)$. See Figure  \ref{fig:diagram:wps} for illustration. 
By deleting $(P_1,a_1)$, the attributes in $P_1\cap a_1$ are the only ones affected and candidates for not satisfying the condition of Definition \ref{dfn-acyclicc}. However, since $P_1\cap a_1=P_1\cap a_2 = \cdots = P_1\cap a_n$, these attributes all appear in the new parse tree in a connected part of it because of edge $(P_1,a_n)$.
\end{proof}

We call the transformation of the parse tree implied by this lemma {\em the reverse path transformation of path  $p$}. 

%

\subsection{Maximal subtrees of a subjoin and other definitions }
\label{sec-23}

We define a {\em partial subtee} of a tree $T$ to be a subtree $T_i$ of $T$  where the leaves of $T_i$ are not necessarily leaves of $T$.

The main theorem  in this paper is the following:

\begin{theorem}
\label{thm-main}
Let $J$ be an acyclic join.
A subjoin of $J$ is safe iff there is a parse tree $T$ of the join $J$ such that the relations in the subjoin form a single partial subtree of $T$.
\end{theorem}

It is convenient to think of a subjoin as a collection of partial subtrees in a specific parse tree of the   join. Thus, we define below maximal subtrees:

\begin{definition}
(Maximal subtrees of a subjoin)

\def\subsetnoteq{\setnoteq{\subseteq}}
Given an acyclic join $J$, a parse tree $T$ of $J$ and a subjoin $J_S$
of $J$, consider the set $S$ of all the relations
participating in subjoin $J_S$; A subset $S_i$ of $S$
is called a maximal subtree of $S$ with respect
to parse tree $T$, if there is a partial subtree of $T$
composed solely of relations in $S_i$, and, there is
no partial subtree of $T$ composed solely of relations in $S'$$ \subseteq S$
where $S_i \subsetneq$$ S'$.
\end{definition}

Given a parse tree $T$ of the join,
we think of the subjoin $Js$ and of its set of relations $S$ as the union of all its maximal
subtrees in $T$, let them be $S_1, S_2, \ldots$ and their roots $R_1, R_2,\ldots$ respectively.

{\sl Observation:} A root of a maximal subtree is neither equal to, nor a child of a node of another maximal subtree. 

This observation is true, because, otherwise, the subtrees are not maximal, since two of them can be viewed as one maximal subtree because they are connected in the parse tree.

%
%

We call {\em subjoin attributes} the attributes that appear in the relations in the subjoin.
We call {\em shared attribute} an attribute that is shared by at least two maximal subtrees in the subjoin.

A relation that belongs to the subjoin is called a {\em subjoin relation}, otherwise it is called an {\em external relation}.
A node of a parse tree whose relation belongs to the subjoin is called a {\em subjoin node}.
Any other node of a parse tree is called an {\em external node.} 
We say that a subjoin node (relation, respectively) $u$ is an {\em associated subjoin node} ({\em associated subjoin relation}, respectively)  of an external node (relation, respectively) $v$ if $u$ contains all subjoin attributes that are contained in $v$. We often say simply  {\em associated  node} or {\em associated  relation}.

Now that we have introduced our terminology we can explain in a technical level the structure of the rest of the paper.
The proof  of Theorem \ref{thm-main} proceeds as follows: We have two  cases, a) when there is an external relation that has no associated relation\footnote{This is equivalent to the property stated in Subsection \ref{subsec-components}} (this is the case in Section \ref{sec-5} and we prove that the subjoin is non-safe in this case), and b) when all external relations have their associated nodes. 

In the second case, we apply repeatedly a procedure  (similar to the one presented in Section \ref{sec-warmup}) that reduces the number of  maximal subtrees. If this procedure fails then we prove
that the conditions of Theorem \ref{thm-sh-attr-cedb} (from Section \ref{sec-warmup}) are satisfied (this is done in Section \ref{sec-counter-shared-sec}), hence, we use Theorem \ref{thm-sh-attr-cedb} to prove that the subjoin is not safe. Section \ref{sec-warmup} considers the special case where we have only two maximal subtrees in the subjoin but it contains many of the complications of the second case which is treated fully in Section \ref{sec-counter-shared-sec}.

\section{Warmup Example: Two Maximal Subtrees}
\label{sec-warmup}

We consider the case where there is a parse tree $T$ of the join $J$ such that the subjoin $J_S$ consists of  two maximal subtrees, let them be $T_1$ and $T_2$. This is the main result of the section:

\begin{theorem}
If the subjoin has two maximal subtrees in a parse tree, $T$, then the following holds: The subjoin is safe if and only if there is a parse tree $T'$ where the the relations in the subjoin form a single  partial subtree of $T'$.
\end{theorem}

The one direction of the above theorem is easy and is presented in the theorem below. The rest of this section proves the other direction.

\begin{theorem}
\label{thm-sh-attr-cedb1}
If there is a parse tree $T$ of the join $J$ such that the relations in the subjoin $J_S$ form a single partial subtree (call it $T_1$) of $T$, then the subjoin $J_S$ is safe. 
\end{theorem}

\begin{proof}
Let $D$ be a fully reduced database. We compute $J_S(D)$ and $J^c_S(D)$. 

If $t\in J_S(D)$ is not a dangling tuple, then there is a tuple $t'\in J^c_S(D)$ that joins with $t$.
Now $t$ is computed from tuples $t_1,\ldots t_m$ of $D$ and $t'$ is computed from tuples $t'_1,\ldots t'_l$ of $D$. Suppose there are two tuples, one from each, i.e., say tuple $t_i$ and tuple $t'_j$ that do not join. Then $t$ and $t'$ do not join either because the projected boundary attributes in $J_S$ and $J^c_S$ span all common attributes in $J_S$ and $J^c_S$.
Hence all pairs of such tuples join.

Suppose  $t\in J_S(D)$  is a dangling tuple. Then, according to the above, there are two tuples of $D$
that do not join. This is contradiction because $D$ is pairwise consistent. 
\end{proof}

\subsection{Structure  of the rest of the section}
\label{subsec-prelimbreak}
To proceed with the proof, we focus on a particular path $p$.

We consider the path, $p$, joining the two roots $R_1$ and $R_2$ of $T_1$  and $T_2$ respectively in the tree $T$; $p$ includes the two roots too.  For the case treated in this section, we assume wlog that the lowest common ancestor of $R_1$ and $R_2$ is neither $R_1$ nor $R_2$ (assuming that a node is also an ancestor of itself)\footnote{otherwise,  we change the root of $T$ to any node of $p$ not in the subjoin}\footnote{For the general case however dealt in Section \ref{sec-counter-shared-sec}, we will have to consider the other case too in order to be technical, although only a simple modification  is needed.}.

We have two cases depending on a property of the path from one root to the other.  In particular, if we delete all shared attributes (between the two maximal subtrees) from this path, then either the path is broken (i.e., there are two consecutive nodes with no common attributes) or not. In the first case we prove that the subjoin is safe and, in the second case, we prove that the subjoin is not safe. We need some definitions first. 

%
%
%
%
%

Let $S$ be the maximal set of attributes that is shared by all nodes of $p$ ($S$ could be empty). 
Hence, $S$ is the set of exactly those attributes shared by both roots $R_1$ and $R_2$.
We consider the partial edges of the hypergraph of $J$  that are generated by $ALL-S$ (where $ALL$ is the set of all attributes in the  join $J$) and refer to the hypergraph thus constructed by $J_p$. We  argue on $J_p$. We refer to the path $p$ after deleting from its nodes the attributes in $S$ (i.e., as it is viewed in $J_p$) as the {\em partial path } $p$.  

\begin{definition}
 \label{dfn-breakk}
We consider the partial path $p$.
We  have two cases : the partial path $p$ is connected or it is disconnected. In the second case we say that there is a {\em break}. We choose two nodes $u$ and $u'$ to define a  {\em break point} as follows: These are nodes $u$ and $u'$ on $p$ that have a child-parent relationship on $T$, such that $u$ and $u'$ do not share any attributes in in patial path $p$.  Wlog, suppose $u$ and $u'$  appears along the path from root $R_1$ to  the least common ancestor (LCA) of the two roots $R_1$ and $R_2$ and  $u'$ is closer to $R_1$ than $u$. 
We say that the pair $(u,u')$ is a {\em break point}  with respect to $T_1$.  
\end{definition}

For an example of a break, see Subsection \ref{subsec-stem-break-examples}.
When there is a break, we use
Proposition \ref{pro-break-reverse}, otherwise we argue as in Subsection \ref{subsec-nobreaktwo}.

%
%
%
%
%
%

\subsection{There is a break}


\begin{proposition}
\label{pro-break-reverse}
Suppose for the acyclic join $J$ we have a parse tree $T$ where the subjoin $J_S$ has two maximal subtrees. Suppose there is a break. Then there is a parse tree where the subjoin $J_S$ has only one maximal subtree. 
\end{proposition}

\begin{proof} 
We consider the break point $(u,u')$. Observe that the path from $u'$ to the root of the  maximal subtree $T_1$ satisfies the shared-attributes condition which is necessary for the reverse path transformation of Subsection \ref{subsec-reverse}.
We apply the reverse path transformation.  Hence, we can obtain a parse tree where the root of $T_1$  is a child of the upper node of the break, $u$. After that, we transform further the parse tree  by transferring the subtree rooted in the root of $T_1$  to be  a child of the root of $T_2$, i.e., technically we only change the parent of the root of $T_1$ to be the root of $T_2$. This last transformation is feasible because only  attributes  in the set $S$ are common between the root of $T_1$  and the upper node $u$ of the break
 and $S$ appears in the root of $T_2$.
\end{proof}

For an example, in Figure \ref{fig:diagram:wps5} in Appendix \ref{app-ex}, the subjoin $AE\bowtie ADE$ is safe because there is another parse tree where 
$AE$ is a leaf again but with parent the node $ADE$.

\subsection{No break. Counterexample database by shared attributes}
\label{subsec-nobreaktwo}

In this subsection, we will prove a more general result than the one needed in the case of two maximal subtrees. We do that because the special case here is not less complicated than the general case treated in Section \ref{sec-counter-shared-sec}.

Considering the  join hypergraph and a set of partial edges generated by a certain set  of attributes, we refer to the join that results from these partial edges (i.e., same schema as these edges) as  {\em partial join}. In 
the same sense we refer to the {\em partial subjoin} of a subjoin. 

The following defines a set of attributes with certain useful properties; we show that such a set exists in the case there is no break.

\begin{definition}
\label{dfn-leads-to}
Let $T$ be a parse tree of join $J$.
Let $T_1$ be a maximal subtree of $J$ in $T$ and
 $B_S$ be a  nonempty set of attributes  with the following property:  Consider the partial edges generated by $B_S$. Then
a) the partial join is connected (as a hypergraph) and 
 b) the partial subjoin is disconnected in the following particular fashion: $T_1$ is disconnected from $\cup_{i\neq 1}T_i$ (i.e., from the rest of the subjoin).

%

Then, we call the set $B_S$ an {\em n-set} with respect to maximal subtree $T_1$  and we say that maximal subtree $T_1$ {\em leads to the n-set $B_S$}. 
\end{definition}

In the case we have two  maximal subtrees,  if there is no break then   it is easy to find an n-set with respect to one of the maximal subtrees, say wlog, wrto $T_1$. We first, delete the maximal set of shared attributes between the two maximal subtrees. Then, we consider $B_S$ to be the set of all attributes that appear on the nodes  of the path (in $T$) from one root to the other root in the parse tree (after the deletion of the shared attributes).  By definition and assumptions made (i.e., there is no break), $B_S$ has the properties as in the definition above. Thus we have found $B_S$ which is a n-set and maximal subtree $T_1$ leads to $B_S$.

We prove the following theorem :  

\begin{theorem}
\label {thm-sh-attr-cedb}
Suppose there is a maximal  subtree $T_1$ and a set of attributes $B_S$ such that the tree
$T_1$  leads to the n-set $B_S$.
Then the subjoin is non safe.
\end{theorem}

\begin{proof}

%
%
We form an imaginary relation with all attributes in the join. We populate it with two tuples.
One tuple has 0 in all attributes. The other tuple has 1 in all attributes in $B_S$, and it has 0  in all other attributes. (Figure \ref{fig:imaginary}).
Now we populate the relations in the  join by the projections of these two tuples. Thus, we build database $D$.


The database $D$ we constructed is fully reduced. This is straightforward by construction.


We consider the partial edges of the hypergraph of $J$  that are generated by $B_S$  and refer to the hypergraph thus constructed by $E$. 
Let database $D_E$  be a database on the schema of $E$  which results from database $D$ after dropping the  attributes (its values actually) that do not appear in $E$.
We use the notation $0|t$ to define a tuple created from tuple $t$ by appending  some 0's.
To continue with the proof of the theorem, we need the two lemmas below, which argue on $E$ and $D_E$.

%
%

\begin{figure}

\begin{center}

\begin{tabular}{c|c|c|c|}
& $ALL-B_S$ & $B_S$\\
\hline
tuple $t_1:$~~~& 000 $\cdots$000  &111 $\cdots$111  \\
\hline
tuple $t_2:$~~~& 000 $\cdots$000  &000 $\cdots$000  \\
\hline
\end{tabular}

\end{center}
\caption{The imaginary relation}
\label{fig:imaginary}
\end{figure}

\begin{figure}

\begin{center}
\begin{tabular}{c|c|c|c|}
& $ALL-B_S$ & $B_S$\\
\hline
relations of type 1, ~~~~single tuple~~~& 000 $\cdots$000  & no attributes  \\
\hline
relations of type 2, tuple of kind 1~~~& 000 $\cdots$000  &111 $\cdots$111  \\
relations of type 2, tuple of kind 2~~~& 000 $\cdots$000  &000 $\cdots$000 \\
\hline
relations of type 3, tuple of kind 1~~~& no attributes  &111 $\cdots$111  \\
relations of type 3, tuple of kind 2~~~& no attributes  &000 $\cdots$000  \\
\hline
\end{tabular}

\end{center}
\caption{The structure of the counterexample database}
\label{fig:structure}
\end{figure}
%
%
%
%
%
%
%
%
%
%
%
%
%
%

\begin{lemma}
Let $J'$  be any subjoin of $J$ and $J'_p$ the partial subjoin of $J'$ with respect to $B_S$.
Consider the constructed database $D$, the partial edges $E$  and the database $D_E$. 

Then the following is true:
A tuple $t$ is in $J'_p(D_E)$ iff the tuple $0|t$  is in $J'(D)$.
\end{lemma}

\begin{proof}
Consider the two disjoint sets of attributes $B_S$ and $ALL-B_S$, where $ALL$ is the set of all attributes in the join. Consider tuples $t_1$ and $t_2$ of the imaginary relation with all attributes (Figure \ref{fig:imaginary}).
For each relation $r_i$ in $J$, there is  a tuple $t_{i1}$   in $r_i$ which is the projection of $t_1$ on the
attributes of $r_i$ and another tuple $t_{i2}$   in $r_i$ which is the projection of $t_2$ on the
attributes of $r_i$.
Hence, each relation $r_i$ of the constructed database $D$ contains either one tuple or two tuples. We have three types of relations illustrated also in Figure  \ref{fig:structure}. More specifically,  each relation of type 1 has only one tuple and 
the value of all attributes in this tuple is 0. Each relation of type 2 has   one tuple with 0's in the
$ALL-B_S$ attributes and 1 in the $B_S$ attributes and another tuple with 0's in the $ALL-B_S$ attributes and 0 in the
$B_S$ attributes (in total two tuples). Each relation of type 3 has also two tuples but its schema consists only of attributes in  $B_S$, one tuple  has all 0's and the other tuple has all 1's tuple.

Thus, in all cases, all the tuples in the relations have  the value 0 in the attributes in $ALL-B_S$  (if $B_S$ is in the schema). Hence, an assignment of values to the attributes for 
computing a tuple in $J'(D)$ corresponds to an assignment of values to the attributes for computing a tuple in $J'_p(D_E)$.
\end{proof}




\begin{lemma}
Consider the partial edges  in $E$ and the database $D_E$. The partial subjoin $J'_p$ (as defined above, with respect to $B_S$) computes a dangling tuple applied on $D_E$ with respect to the partial complement subjoin $J'_{pc}$.\footnote{In more detail, $J'_{pc}$ is the partial subjoin with respect to $B_S$ of the complement subjoin $J'_c$ of $J'$.}
\end{lemma}

\begin{proof}

The proof is based on the following remarks:  In database $D_E$, some relations have one tuple and some relations have two tuples (see  Figures \ref{fig:imaginary} and \ref{fig:structure} for an illustration too). The relations with two tuples are the ones that have at least one attribute from $B_S$. These relations  form a set which has two disjoint subsets, one subset being part of $T_1$ and the other subset being part of of the rest of the subjoin. We have already pointed out that $B_S$ is the set of attributes on connected relations in $J$, hence the join  $J'_{pc}(D_E)$ contains only two tuples, in particular the ones that have either all 0s or all 1s in attributes in $B_S$. However, the set of attributes from $B_S$ that appear in $T_1$ is disjoint from the set of attributes from $B_S$ that appear in the rest of the subjoin. Hence when we compute the subjoin $J'_p(D_E)$, we have a Cartesian product. This means that there is a tuple computed that have necessarily both 1s and 0s in attributes in $B_S$. This tuple is dangling because it cannot join with any tuple in $J'_{pc}(D_E)$. The reason is that any tuple in $J'_{pc}(D_E)$ has either all 0's or all 1's in the attributes of $B_S$ because all the attributes in $B_S$ are connected in the hypergraph of $J'_{pc}$.
\end{proof}\end{proof}

\section{There is an External Node that does not have an Associated Subjoin Node}
\label{sec-5}


The main theorem of this section is the following:
\begin{theorem}
\label{thm-chase}
Let $T$ be a parse tree of the join. 
If there is an external node in $T$ that does not have an associated node in the subjoin, then the subjoin is not safe.
\end{theorem}

The high level description of the algorithm that constructs the counterexample database is: a) we define a set of child-to-parent and parent-to-child tuple generating dependencies (tgds, for short) b) we construct a seed database by populating the relations in the join with some tuple and c) we apply the chase algorithm on the seed database using the tgds we constructed in order to construct finally the counterexample database. 

\subsection{Construct the child-to-parent and parent-to-child tgds}

{\em Tuple generating dependencies (tgd's for short) } that we use here are first order formulas of the form\footnote{their definition is more general than that, but we do not need it here} 
$$r_i(x_1, \ldots, x_n, y_1, \ldots, y_k)\rightarrow r_j(y_1, \ldots, y_k,z_1,\ldots, z_m)$$ where $r_i$ and $r_j$ are relations and the 
$x_l$'s, $y_l$'s and $z_l$'s are variables that represent their attributes. We call the $z_l$'s {\em existentially quantified} variables.
We say that such a tgd is {\em satisfied} in a database instance if whenever there is a tuple $(a_1, \ldots, a_n, b_1, \ldots, b_k)$ (where $a_1, \ldots, a_n, b_1, \ldots, b_k$ are constant values) in the relation $r_i$ then there are constant values $c_1,\ldots, c_m$ such that there is a tuple
$(b_1, \ldots, b_k,c_1,\ldots, c_m)$  in the relation $r_j$. 

A  {\em chase step} considers a tgd like the above and if there is a tuple $(a_1, \ldots, a_n, b_1, \ldots, b_k)$ (where $a_1, \ldots, a_n, b_1, \ldots, b_k$ are constant values) in the relation $r_i$ and there are no constant values $c_1,\ldots, c_m$ such that there is a tuple
$(b_1, \ldots, b_k,c_1,\ldots, c_m)$ in the relation $r_j$ we do as follows:  We add a tuple $(b_1, \ldots, b_k,c_1,\ldots, c_m)$ in the relation $r_j$, where
$c_1,\ldots, c_m$  are distinct fresh constant values that have not appeared before in the database instance. 

When  the chase algorithm is described in the literature,  labelled nulls are used used instead of distinct fresh constants. Here, we have chosen to replace them by fresh constants in order to keep the terminology simple, since it does not make any difference as long as the fresh constant chosen (arbitrarily) is different from any other constant in the database instance.

The algorithm {\em chase} is a series of chase steps. We say that the chase {\em terminates} if there no more chase steps to be applied, i.e., the tgds are satisfied on the database created by the chase algorithm.

We consider a  parse tree, $T$, of the join.
Suppose $r$ is a parent and $r'$ is one of its children on $T$. For this pair of nodes of the parse tree we construct two tgds:
The child-to-parent tgd is of the form $r' \rightarrow r$ and the parent-to-child tgd is of the form
$r \rightarrow r'$. In both, the  attributes/variables shared between the two nodes of the parse tree (the ones that represent the child and its parent) are the same on both sides of the tgd, while the nonshared variables are existentially quantified in the child-to-parent tgd when they belong only to the parent and, in the parent-to-child tgd when they belong only to the child.
More specifically, we define a parent-to child tgd to be: $$r(X_1, X_2, \ldots, Y_1, Y_2,\ldots)\rightarrow r'(Y_1, Y_2,\ldots, Z_1, Z_2, \ldots )$$
where $r$ is the parent of $r'$ and the $X_i$s belong only to the parent whereas the $Y_i$s belong to both $r$ and $r'$, and the $Z_i$s belong only to the child. Without loss of generality, we assume that the $Y_i$s appears in the first positions in $r'$ and in the last positions in $r$. Similarly we define a child-to-parent tgd, only now the child appears on the left hand side (lhs, for short) and the parent on the right hand side (rhs, for short) of the tgd. We form such tgds for each pair of child-parent on the parse tree. This is the set $\Sigma$ of tgds that we will use. The set $\Sigma$ is not unique to the join, it depends on the parse tree considered.

\subsection{Construct the counterexample database instance by chase}
\label{subsec-counter-db-chase}
We use the above constructed set $\Sigma$ of tgds and chase  with $\Sigma$   a seed database instance (that we will construct shortly) to build the database instance which will serve as proof that the subjoin is not safe for this case, i.e., we form the counterexample database. Specifically, we do as follows:

\begin{itemize}
\item
First, we construct the {\bf seed database instance} as follows: We add a single tuple in each subjoin relation. This tuple is created as follows: We imagine that we have a {\em subjoin seed tuple} $t$ (chosen arbitrarily) on the attributes of the subjoin and populate each subjoin relation with one {\em seed tuple} which is the projection of $t$ on the attributes of the specific relation we are populating. The values that we chose in the subjoin seed tuple are called 
{\em seed values}.
\item
Then we use the child-to-parent and parent-to-child set $\Sigma$ of tgds and apply the chase algorithm.
\end{itemize}

 Example  \ref{ex-chase}  illustrates  the construction of the tgds as well as and the construction of the counterexample database.

\begin{example}
\label{ex-chase}

We consider the join $J=ABC\bowtie AB \bowtie AC \bowtie BC$ and the parse tree $T$ which has root the relation $ABC$ and there are three children nodes of the root, which are the rest of the relations. We consider the subjoin $J_S= AB \bowtie AC \bowtie BC$.
First we observe that the relation $ABC$ is an external relation which does not have an associated relation in the subjoin,  because none of the three relations in the subjoin contains all attributes that are contained in the relation $ABC$.

Now, we construct the tgds in $\Sigma$, assuming that relation $r_0$ is $ABC$ and relations
$r_1,r_2,r_3$ are the $AB,AC,BC$ respectively.

~~~~~~~~$d_1:~~r_1(x,y)\rightarrow r_0(x,y,z)$, ~~~~~$d_4:~~r_0(x,y,z) \rightarrow r_1(x,y)$

~~~~~~~~$d_2:~~r_2(x,z)\rightarrow r_0(x,y,z)$, ~~~~~$d_5:~~r_0(x,y,z) \rightarrow r_2(x,z)$

~~~~~~~~$d_3:~~r_3(y,z)\rightarrow r_0(x,y,z)$,  ~~~~~$d_6:~~r_0(x,y,z) \rightarrow r_3(y,z)$,

The seed database (assuming we start with subjoin seed tuple $(a,b,c)$) is:
The relation $r_1=AB$ contains the tuple $(a,b)$, the relation $r_2=AC$ contains the tuple $(a,c)$ and the relation $r_3=BC$ contains the tuple $(b,c)$.

Now we apply three chase steps  using tgds  $d_1,d_2$ and $d_3$ and populate the relation $r_0=ABC$ with the following tuples: $\{  (a,b,c_1), (a,b_1,c), (a_1,b,c)   \}$.
Next, we apply three chase steps  using tgds $d_4,d_5$ and $d_6$ and populate the relations $r_1, r_2, r_3$ with more tuples as follows: We add to $r_1$ two tuples
$\{ (a,b_1), (a_1,b) \}$, we add to $r_2$ two tuples
$\{ (a,c_1), (a_1,c) \}$, and we add to $r_1$ two tuples
$\{ (b,c_1), (b_1,c) \}$. That completes the construction of the counterexample database $D$.  Notice that it is the same database instance as the one we discussed in Example \ref{ex-begin}.
\end{example}

\subsection{Proof that {$\mathbf D$} is indeed a counterexample database}

We will show now that  the chase terminates and produces a database instance which is fully reduced.
\begin{theorem}
Consider an acyclic join $J$,  a parse tree $T$ of $J$ and the set $\Sigma$ of tgds constructed as in Subsection \ref{subsec-counter-db-chase}.
Then the chase using $\Sigma$  terminates when applied  on the seed database instance  and the database instance $D$ that is produced is fully reduced.
\end{theorem}

\begin{proof}
It is convenient to argue about termination if we apply the chase in a certain order.  We apply chase in two phases, one phase upwards in the parse tree and one phase downwards as follows:  In the upwards phase, we apply the child-to-parent tgds bottom up. In the downwards phase, we apply the parent-to-child tgds top down. We will prove that this two-phase chase produces a database on which  all tgds in $\Sigma$ are satisfied.

Inductively, suppose the chase terminates on a  parse tree with less than $n$ nodes. Now, consider a parse tree, $T$,  with $n$ nodes. 
In the upwards phase of the chase, the root of $T$ is populated with some tuples because of child-to-parent tgds  with its children, hence these tgds are now satisfied. In the downwards phase of the chase, the children of the root are populated with some tuples, hence the parent-to-child tgds with its children are satisfied and the extra tuples do not trigger dissatisfaction of child-to-parent tgds because they are  produced only from the tuples of one node (the parent) and they all satisfy the parent-to-child tgd, by construction of the tgds (notice the symmetry between the two tgds of the same pair of nodes).  The chase terminates on the subtrees rooted at the children of the root, by inductive hypothesis, hence it terminates on the parse tree with $n$ nodes as well.

Now we need to prove that, if the tgds in $\Sigma$ are satisfied on database $D$, then $D$ is fully reduced. 
We use Theorem \ref{thm-semijoinprogram} and {\bf Procedure} {\sl Semijoin}.  
We will prove that 
 the  {\bf Procedure} {\sl Semijoin} does not delete any tuples in $D$. 

We argue recursively on the parse tree $T$.
Let $D$ be a database instance of relations on a specific partial subtree $T'$ of $T$. 
Let $r$  be a leaf relation in $T'$. 
Recursively suppose  database $D'=D-r$ is fully reduced.  


Suppose relation $r$ has a dangling tuple in $D$ (which the semijoin procedure will delete in its downwards phase). This means however that the specific tgd with the parent of $r$ is not satisfied. 
Suppose the parent of $r$ has a dangling tuple. In this case the tgd with respect to its parent is not satisfied. 
Hence, $D$ is fully reduced too.
%
\end{proof}


We have proven that $D$ is fully reduced.
Now, it remains to be proven  that output of the subjoin on $D$ has a dangling tuple with respect to the output of the complement of the subjoin on $D$ (i.e., $J_S(D)$ has a dangling tuples with respect to $J^c_S(D)$). This is a straightforward consequence of the following theorem:

%
%
%
%
%
%

\begin{theorem}   
Consider an acyclic join $J$,  a parse tree $T$ of $J$ and the set $\Sigma$ of tgds constructed as in Subsection \ref{subsec-counter-db-chase}.
The  chase using $\Sigma$  when applied  on the seed database instance produces a database instance $D$ for which the following is true:

The output of the subjoin on $D$ includes the seed relation tuple projected on its output attributes but the output of the complement  subjoin on $D$  does not include a tuple whose projection on the boundary attributes   is the seed relation tuple projected on these boundary attributes.
\end{theorem}

\begin{proof}
When the first chase step is applied then the relations/nodes that are populated with a tuple where all the boundary attributes have seed values then this means that this node has an associated node which is the node which was used for this chase step. Iteratively, this is the case for each node when the $i$-th step is applied. Since there is a node with no associated subjoin node, this node has all its tuples with at least one  boundary attribute having a non-seed value.  Hence, in each tuple of the output of the complement subjoin there is at least one subjoin attribute that has a non-seed value.
%
%
%
%
%
%
\end{proof}


%

\section{All External Nodes Have Associated Nodes in the Subjoin 
}
\label{sec-counter-shared-sec}

Now we assume that, for every external node $u$, there is at least one subjoin node that contains all the subjoin attributes of $u$. Remember, we call such a subjoin node an associated node  of $u$.
Each external node may have multiple associated nodes.

This section describes one iteration in the case where all external nodes have associated nodes in the subjoin. It considers as input a subjoin and a parse tree and in the output, either a decision is made that the subjoin is not safe or,  it outputs a different parse tree, on which the subjoin  has strictly fewer maximal subtrees than the parse tree in the input. The next subsection presents some definitions.

\subsection{Lowest maximal subtrees, stems, siblings}
\label{subsec-pick-subtree}
The following definition allows for a convenient property. Informally,  this property allows a simple transformation of the parse tree by moving the chosen maximal subtree to another position without ``carrying '' with it other maximal subtrees and thus introducing complications unnecessarily. 

\begin{definition}({\bf Lowest maximal subtree})
A {\em lowest maximal subtree} is a maximal subtree  such that it has no node with a descendant that is a root of another
maximal subtree.
\end{definition}

Proposition \ref{pro-not-disturb} states that a lowest maximal subtree exists.

\begin{proposition}
\label{pro-not-disturb}A maximal subtree, $T$, with the greatest depth is a lowest
maximal subtree.
\end{proposition}

\begin{proof}
Suppose $T$ is not a lowest
maximal subtree. Then a maximal subtree exists which has
 a root that is a descendant of a node of $T$. Hence, it  has depth greater than
the depth of the root of $T$. This is a contradiction, since we chose $T$ to have the greatest depth of its root to the
root of the whole tree.
\end{proof}

\begin{definition}({\bf Stem})
\label{dfn-stem}
Let $T$ be a lowest maximal subtree.
The path from $R$, the root of $T$, to the root of the whole tree has a node $v$
which is the uppermost node
that has the property: the part of the
path, call it $p$, from $R$ to $v$ is such that every node of $p$ has no descendant that is a root of a maximal subtree.

This path $p$ is called the {\em stem} of $T$. Node $v$ is called the {\em upper tip}, or simply {\em tip}  of the stem. The root of $T$ and $v$ are the {\em endpoints} of the stem. See   Appendix \ref{app-ex} for an example.

\end{definition}
The definition of a break is the same as  Definition \ref{dfn-breakk} where path $p$ is a stem of a lowest maximal subtree and it becomes partial path $p$ be deleting all shared attributes.

Notice that the upper tip of a stem falls in one of the following two cases:

(i) It is a node of another maximal subtree $T_a$. In this case we say that $T$ is {\em hanging} from $T_a$. We call this maximal subtree {\em dependant}.

(ii) It is an external node.  We call this maximal subtree {\em not dependant}.


The upper tip of a stem (and, hence the stem) can be equivalently defined as the lowest common ansector (LCA), over all other maximal subtrees,  of the root of the tree under consideration and another maximal subtree.

\begin{definition}({\bf Siblings})
Two lowest not dependant maximal subtrees that have the same upper tip of their stems are callled {\em siblings}.
\end{definition}

The following proposition states that when the upper tip of a stem is an external node, then
we can always find two siblings.

\begin{proposition}
\label{pro-AB}
Suppose there no maximal subtrees that are dependant.
Suppose there exists a lowest maximal subtree whose upper tip of the stem is an external node.
Suppose the subjoin has at least two maximal subtrees. Then there are at least two lowest maximal subtrees $T_1$ and $T_2$ that are siblings.
\end{proposition}
\begin{proof}
Consider the stem with the lowest upper tip of the stem (i.e., this upper tip is at the greatest depth from the  root of the whole parse tree); call this tip $S_e$.
Suppose there is no stem with its upper tip on $S_e$. Since, there are certainly (otherwise $S_e$ would not have been
ended there but had to go higher in the parse tree) maximal subtrees with roots being descendant of $S_e$, their stem tip should be lower than $S_e$ and, hence, have an
upper tip lower than $S_e$, this is a contradiction.
\end{proof}

%
%
%

\subsection{Proof of the main result of this section}
\label{subsection-dfnBreaks}

Before we proceed, we state the following theorem whose proof is the same as the proof of Proposition \ref{pro-break-reverse} with the only difference that now the parent of the root of tree $T_1 $ is the associated node to the  node which defines a break point.

\begin{theorem}
\label{thm-break-general}
Suppose there is a break in the stem of maximal subtree $T_1$ and suppose that the node of the break point has an associated node in a maximal subtree other than $T_1$. Then, we build a parse tree with strictly fewer maximal subtrees.
\end{theorem}

As we mentioned, a lowest maximal subtree is either dependent or not. The following two theorems prove the main result in this section by considering each of these cases. The proofs of the two theorems have many similarities, so we move the proof of Theorem \ref{thm-counter-good2} in Appedix \ref{app-last}.


\begin{theorem}
\label{thm-counter-good1}
Suppose the subjoin has more than one maximal subtree. 
 Let $T_1$ be a lowest maximal subtree which is dependant. Then either there is a parse tree with strictly fewer 
maximal subtrees in the subjoin  or $T_1$ leads to an n-set $B_S$.
\end{theorem}

\begin{proof}
%
%
%

Suppose
$T_1$ is hanging from a maximal subtree, let it be $T_2$.
The shared attributes between $T_1$ and $T_2$ are all the attributes that are shared by 
$T_1$ and the rest of the join. Hence, after removing them (and considering the partial hypergraph edges generated),  $T_1$ does not share any attributes with the rest of the subjoin. 

Suppose we have a break in the stem of $T_1$. Any node of the stem of $T_1$ has an associated node in $T_2$ (this node is the upper tip of the stem of $T_1$). Hence the parent of the root of $T_1$ will be the upper tip of the stem of $T_1$, according to Theorem \ref{thm-break-general}.

Suppose there is no break in the stem of $T_1$.
Moreover, 
 the root of $T_1$ and $T_2$ are connected when considering the partial edges otherwise we would have a break.  Hence, if we consider as $B_S$ the set of attributes in the stem of $T_1$ after the removal of the shared attributes,   we observe that set $B_S$  has the properties of  Definition \ref{dfn-leads-to}.  
\end{proof}

\begin{theorem}
\label{thm-counter-good2}
Suppose the subjoin has more than one maximal subtree. 
 Let $T_1$ be a lowest maximal subtree which is non dependant. Then either there is a parse tree with strictly fewer 
maximal subtrees in the subjoin  or $T_1$ leads to an n-set $B_S$.
\end{theorem}

The two above theorems and Theorem \ref{thm-sh-attr-cedb} lead, in a straightforward way, to the following theorem which is the main result of this section:
\begin{theorem}
\label{thm-counter-good}
Suppose all external nodes have associated subjoin nodes.
Suppose the subjoin has more than one maximal subtree. 
Then either  there is a parse tree with strictly fewer 
maximal subtrees in the subjoin  or the subjoin is not safe.
%
%
\end{theorem}

\section{Proof of the Main Theorem \ref{thm-main}}
We have two cases: 

a) There is an external relation with no associated subjoin relation. Then the subjoin is not safe according to Theorem \ref{thm-chase}.

b) All external relations have associated subjoin relations. Let $T$ be a parse tree with minimal number of  maximal subtrees. If $T$ has only one maximal subtree then the subjoin is safe according to Theorem \ref{thm-sh-attr-cedb1}. 
Otherwise the subjoin is not safe according to Theorem  \ref{thm-counter-good}.

%
%


\section{Subjoin-optimal Parse Tree}
\label{sec-min}
We present here an algorithm which, given an acyclic join and a subjoin, finds a parse tree of the join with the minimum number of maximal subtrees in the subjoin.

\subsection{Algorithm}
Let $p$ be a path in parse tree $T$ of acyclic join $J$.
Let $S$ be the maximal set of attributes that appear in all nodes of $p$ ($S$ could be empty). 
We consider the partial edges of the hypergraph of $J$  that are generated by $ALL-S$ (where $ALL$ is the set of all attributes in the  join $J$) and refer to the hypergraph thus constructed by $J_p$. 
We refer to the path $p$ after deleting from its nodes the attributes in $S$ (i.e., as it is viewed in $J_p$) as the {\em partial path } $p$.  

\begin{definition}
 \label{dfn-general breakk}
We consider a  partial path $p$ whose endpoints are nodes of two different maximal subtrees,  $T_1$ and $T_2$, of $T$ and all other nodes are non-subjoin nodes. 
We  have two cases: the partial path $p$ is connected or it is disconnected. In the second case we say that there is a {\em general break with respect to $T_1$ and $T_2$}. 
\end{definition}
When there is  a general break, then we  choose two nodes $u$ and $u'$ to define a  {\em general break point} as follows: These are nodes $u$ and $u'$ on $p$ that have a child-parent relationship on $T$, such that $u$ and $u'$ do not share any attributes in  patial path $p$.  
We say that the pair $(u,u')$ is a {\em general break point} wrto $T_1$ and $T_2$. 

Let $T$ be a parse tree of acyclic join $J$. 
An {\em arc} in $T$ joins a node of $T$ to its parent.
Let $e$ be an arc not in $T$. Let $e'$ be an arc in $T$ such that if both considered in $T$, there is a cycle containing both.
We produce parse tree $T_1$ which results form $T$ after adding $e$ and deleting $e'$. 
We define a {\em change} to be such a pair $(add, delete)$.
%

When there is a general break, we apply the following algorithm to obtain a parse tree when the subjoin has fewer maximal subtrees than in the original parse tree. 

{\bf Algorithm:}
%
%
%

Suppose there is a gneral break in given parse tree $T$ wrto $T_1$ and $T_2$. 
 We produce parse tree $T'$ from $T$ as follows: We use  the  reverse path transformation from Section
\ref{subsec-reverse}. We first apply the reverse path transformation  considering a maximal subtree $T_1$ and suppose  we delete arc $e_1$ and add $e'_1$, according to this transformation. Then we apply the operation of having the root of the maximal subtree $T_1$ as a child to one of the nodes of  maximal subtree $T_2$, i.e., we delete $e'_1$ and add another arc $e_2$ appropriately. This can be described as one change, i.e., we delete $e_1$  and add $e_2$.  We replace $T$ with $T'$ and repeat. We stop when there is no general break.

In Appendix \ref{app-newsec} we prove the following theorem:

\begin{theorem}
The algorithm always produces a parse tree with minimum number of maximal subtrees.
\end{theorem}

\appendix

\section{Semijoins and fully reduced database}
\label{sec-semijoin}
A {\em semijoin statement} denotes a semijoin between two relations and is written $r_i\ltimes r_j$ where $r_i$ and $r_j$ are relations.
A {\em semijoin program} is a linear sequence of semijoin statements. This is an example of  a semijoin program:
\def\defeq{\mathrel{\mathop:}=}
\label{ex-semi}
$r_3\defeq r_3 \ltimes r_4;$
$~~~r_3\defeq r_3 \ltimes r_5;$
$~~~r_2\defeq r_2 \ltimes r_1.$

\begin{definition}
We say that a semijoin program {\em fully reduces} a database instance, if after the program is applied on the database instance, we obtain a new database instance which is
 consistent and produces the same output as when $J$ is applied on the original database instance.
\end{definition}

\begin{theorem}
\label{thm-semijoinprogram}
(\cite{BeeriFMY83} \cite{BernsteinG81})
If $J$ is an acyclic join, then there is a semijoin program that fully reduces every database instance.
\end{theorem}

The program of the above theorem consists of $2n-2$ semijoin statements where $n$ is the number of relations in the join and is the following procedure:

{\bf Procedure} {\sl Semijoin}

Input: acyclic join $J$ and a parse tree $T$ for $J$; a database instance $D$.

Step 1: Let relation $r$  be a leaf of $T$ and $r'$ be the parent of $r$. Do
$r'\defeq r' \ltimes r$.

Step 2: Recursively we generate a consistent database $D'$ out of $D-r$ (i.e., the database that results from $D$ after we remove relation $r$).

Step 3: Add to $D'$ the semijoin $r\defeq r \ltimes r'$.

Theorem \ref{thm-semijoinprogram} is an ``if and only if,'' i.e., the following theorem is true:


\begin{theorem}
\label{thm-guarantee}
(\cite{BeeriFMY83} \cite{BernsteinG81})
If $J$ is  not an acyclic join, then no semijoin program is guaranteed to fully reduce  all relations in any database instance.
\end{theorem}

A {\em join expression} is a parenthesization of binary joins such that
$(r_1 \bowtie r_2) \bowtie (r_6 \bowtie r_3)$.

{\em Monotone join expression with respect to  a database instance:}
If every binary join that appears in the expression is over consistent relation instances.
A {\em Monotone join expression} is one that is monotone with respect to every pairwise consistent database instance.
%
%
%
A join is acyclic iff there is a monotone join expression.

\section{Examples}
\label{app-ex}
\subsection{Parse trees and safe subjoins}
Figure \ref{fig:diagram:wps5} shows an example of a parse tree.

\begin{figure}[ht]
\begin{center}
\includegraphics[height=19cm, angle=0, keepaspectratio=true,clip]{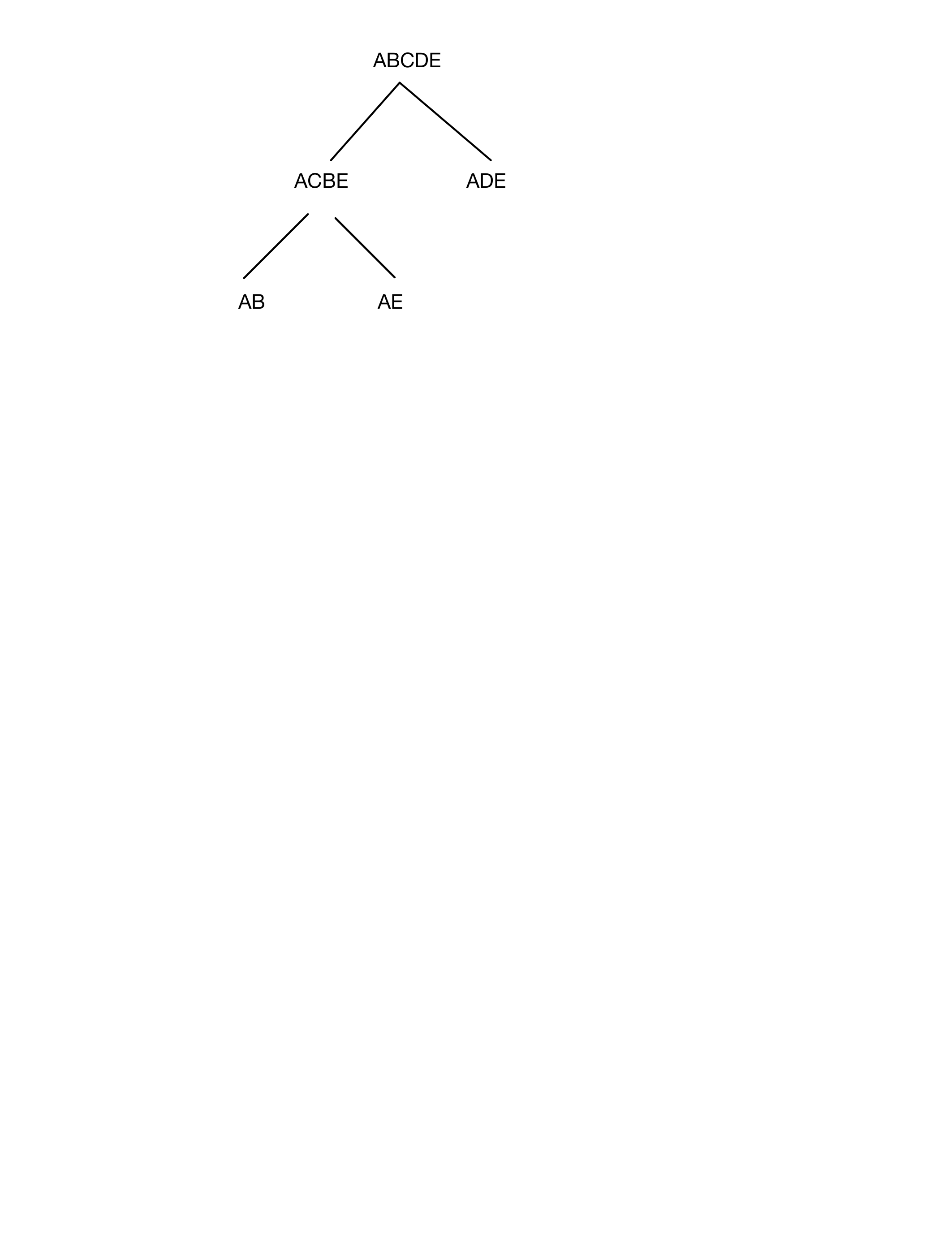}
\end{center}
\vspace*{-14.5cm}
\caption{A parse tree for the join $ABCDE\bowtie ACBE \bowtie  ADE  \bowtie AB \bowtie AE$. This join is acyclic. \label{fig:diagram:wps5}}
\end{figure}

In the following example we show that a nonsafe subjoin can be acyclic.
\begin{example}
\label{ex-acyclic-non-safe}
Consider the join  $$J=ABCDE\bowtie ACBE \bowtie  ADE  \bowtie AB \bowtie AEF\bowtie ABCDEF$$ Join $J$ is acyclic and contains as subjoin the following:
$$J_S=ABCDE\bowtie ACBE \bowtie  ADE  \bowtie AB \bowtie AEF$$
 The subjoin $J_S$ is acyclic with same parse tree (just replace leaf $AE$ with $AEF$) as in  Figure \ref{fig:diagram:wps5}. However the subjoin $J_S$
is not safe. This is proof that a non-safe subjoin can be acyclic.
\end{example}

\subsection{Joins with Projections}
\label{sec-projection}
The following examples shows how we find the attributes that are projected to the ouput of the subjoin and the complement subjoin.
\begin{example}
We consider the join $J=ABC\bowtie AB \bowtie ACE \bowtie BCF\bowtie FG$. This is an acyclic join. We project out in the output  the attributes $A$ and $E$. Now, we consider the subjoin $J_S=BCF \bowtie FG$. The projected attributes in the ouput of this subjoin is only the attribute $A$ because it appears in the output of $J$  and the boundary attributes, which are $B$ and $C$. 
Notice that  the other projected attribute, $E$, does not appear in the subjoin.

The complement subjoin is $J^c_S=ABC\bowtie AB \bowtie ACE$ and the attributes projected in the output are
$A,E$ and the boundary attributes $B$ and $C$. 
\end{example}

\subsection{Examples on Section  \ref{sec-warmup}}

Now we further elaborate on the example of Figure \ref{fig:diagram:wps5} to illustrate the arguments and results in Section  \ref{sec-warmup}. 

\begin{itemize}      

\item  In Figure \ref{fig:safe1}, we have listed all subjoins  that contain only two relations of the join in Figure \ref{fig:diagram:wps5}. We have uded  two columns, the first column contains the safe subjoins and the second column the nonsafe subjoins among those.

\item 
We consider the nonsafe subjoin from this list, $AB\bowtie AE$ and show that it is not nonsafe by constructing the counterexample database
that we described in Section  \ref{sec-warmup}. 
Notice that the shared attributes is only one, the attribute $A$.
We have two maximal subtrees here, each being one relation, and the attributes $B$ and $E$ comprise the set $B_S$ mentioned in the proof. 
Thus, the imaginary relation  $ABCDE$ contains the two tuples  $(00000)$ and $(01001)$. 
The two relations in the subjoin have the following tuples in the counterexample database: The relation $AB$ has the tuples $(01)$ and $(00)$ and the relation $AE$ has  the tuples $(01)$ and $(00)$. Thus, the  subjoin $AB\bowtie AE$ contains four tuples.

\item
We consider the nonsafe subjoin from the same list, $AB\bowtie ADE$. 
Notice that the shared attributes is only one, the attribute $A$.
Now the attributes in the set $B_S$ mentioned in the proof, are $B$, $D$ and $E$. Thus the imaginary relation $ABCDE$ contains the two tuples 
$(00000)$ and $(01011)$.
\end{itemize}

\begin{figure}

\begin{center}

\begin{tabular}{|c|c|c|}
safe subjoins & nonsafe subjoins\\
\hline
$AE \bowtie ADE $&$AB\bowtie ADE $ \\
\hline
 $ACBE\bowtie AB$& $ACBE \bowtie ADE$ \\
\hline
 $ACBE\bowtie AE$& $AB \bowtie AE$ \\
\hline
\end{tabular}

\end{center}
\caption{The subjoins with two relations of the join in Figure \ref{fig:diagram:wps5} categorized according to being safe or nonsafe.}
\label{fig:safe1}
\end{figure}

\subsection{Examples of the concept of stem and the concept of break}
\label{subsec-stem-break-examples}

{\sl Illustrating the concept of break:}
In figure \ref{fig:break}, we have two parse trees of two different joins. Both joins (and their corresponding parse trees) have three relations/nodes. The black nodes represent the subjoin considered in each case. So both have a subjoin with two nodes. Both subjoins have two maximal subtrees (their two nodes). The stem of the lowest maximal subtree consists of all three nodes shown. In both cases we have two shared attributes, the attribute $A$ and the attribute $B$. Now, in the case (a), there is no break because, if we delete $A$ and $B$ we are left with a path $(CD),(CE),(E)$ which is connected as a hypergraph path. In case (b), however, we are left with a path $(CD),(E),(EF)$ which is disconnected because nodes $(CD)$ and $(E)$  do not share an attribute (the sets $\{C,D\}$ and $ \{E\}$ are disjoint). Hence there is a break $(~ (ABE), (ABCD))~$ and the break point is the node labeled $ABE$. 

As a consequence of the break, in case (b), we can create another parse tree where the node labeled $ABEF$ is a child of the node labeled $ABCD$ and the node labeled $ABE$ is a child of the node labeled $ABEF$.

\begin{figure}[h]
\begin{center}
\includegraphics[height=19cm, angle=0, keepaspectratio=true,clip]{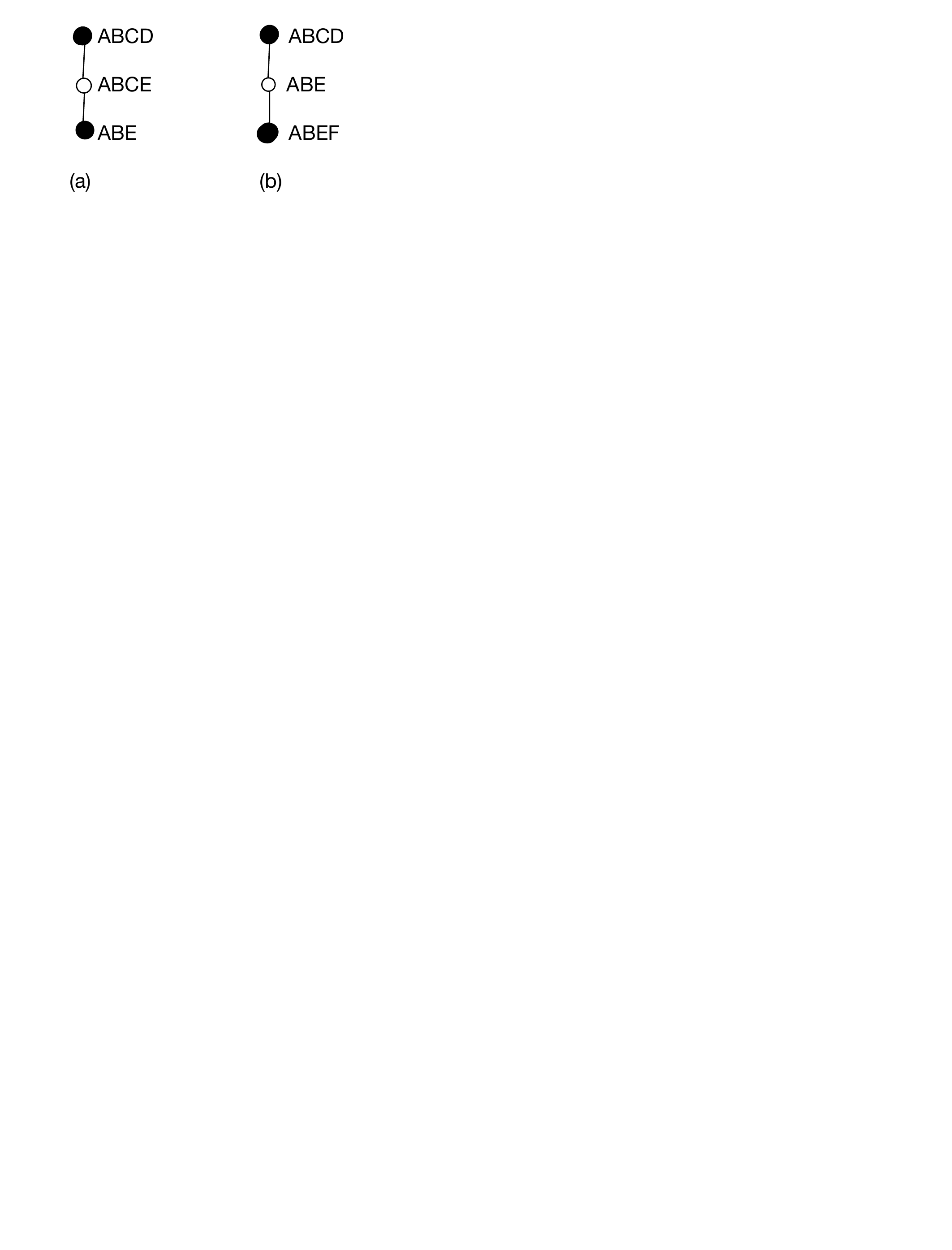}
\end{center}
\vspace*{-15.9cm}
\caption{ Illustrating the concept of {\sl break}. (a) has no break, (b) has a break.\label{fig:break}}
\end{figure}

\begin{figure}[h]
\vspace*{-3.9cm}
\begin{center}
\includegraphics[height=14cm, angle=0, keepaspectratio=true,clip]{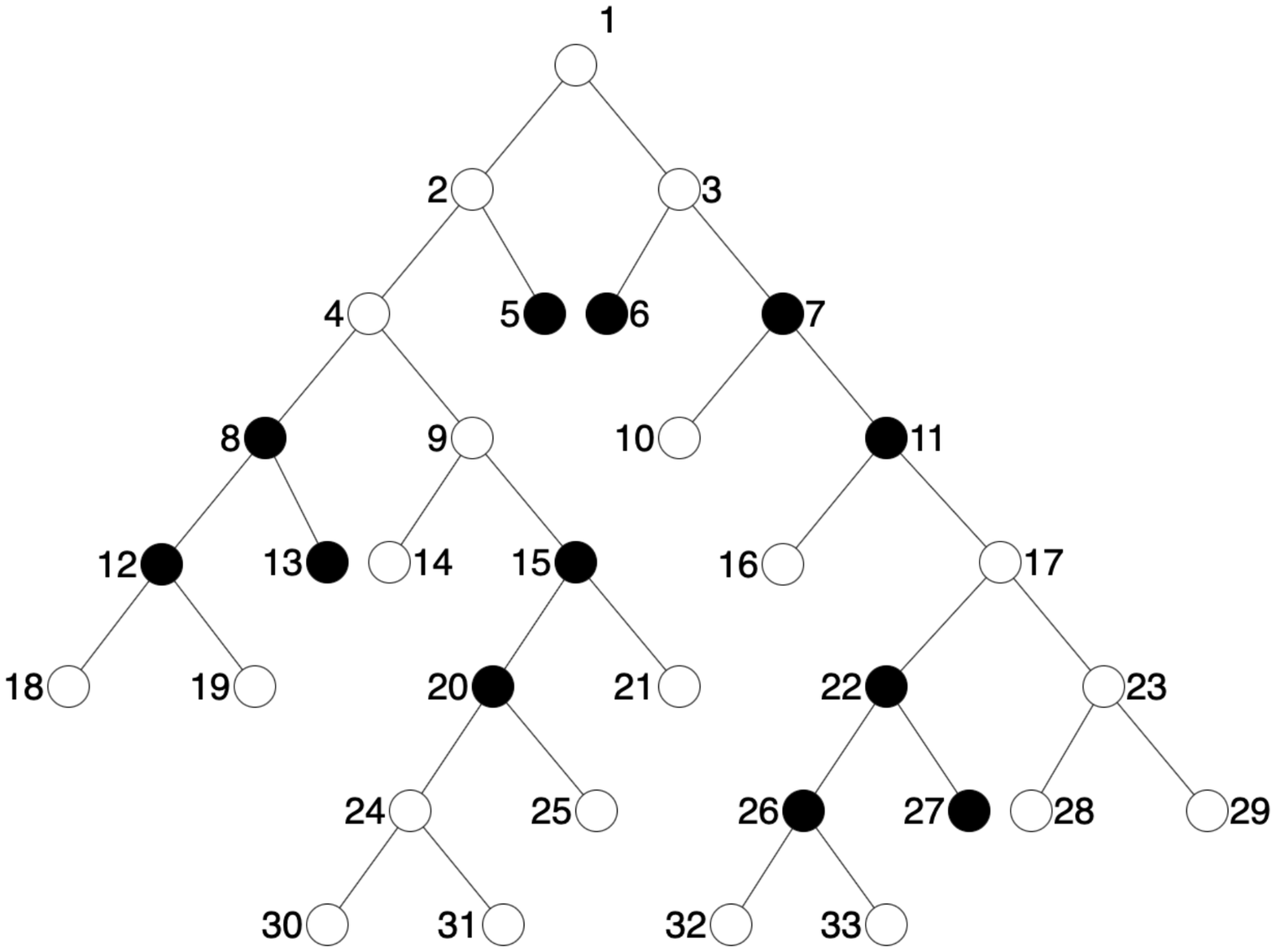}
\end{center}
\vspace*{-2.9cm}
\caption{ Illustrating the concept of {\sl stem}.\label{fig:tree}}
\end{figure}

{\sl Illustrating the concept of stem:}
In figure \ref{fig:tree}, 
we have a parse tree of some acyclic join $J$\footnote{Here the parse tree is binary, but,  in general a parse tree is not necessarily binary.
}   The  subjoin, $J_s$,  under consideration is marked with the
black nodes. There are six maximal subtrees which we list here: $T_1=\{ 5 \}$, $T_2=\{  6\}$, $T_3=\{ 7,11 \}$, $T_4=\{15,20  \}$, $T_5=\{8,12,13  \}$, $T_6=\{ 22,26,27 \}$. All are lowest maximal subtrees except $T_3$. $T_3$ is not a 
lowest maximal subtrees either of its nodes has a descendant that is the root of another maximal subtree (here it is node 22 which is the root of maximal subtree $T_6$.

For tree $T_5=\{8,12,13  \}$, the stem is $(8,4)$; it is not $(8,4,2)$ because node 4 has already a descendant that belongs to the subjoin (e.g., node 15).  For tree $T_4=\{15,20  \} $ the stem is $(15,9,4)$. For tree $T_6=\{ 22,26,27 \}$ the stem is $(22,17,11)$. 

Finally, we consider a third example to see both concepts of break and stem which are demonstrated in a trivial manner in the specific example. We consider a join which is a star, i.e., it is $J_{star}=R_1(A,B_1) \bowtie R_2(A,B_2) \cdots 
\bowtie R_n(A,B_n)$. It is easy to see that any tree with $n$ nodes can be the underlying tree for a parse tree of the join (notice the symmetry of the relations).  Thus, all subjoins of $J_{star}$  are safe.  Moreover, for any stem, there is a break.

\section{Proof of
Theorem \ref{thm-counter-good2}} 
\label{app-last}.
\begin{proof}

Here the upper tip of the stem of $T_1$ is an external node, hence, according to  Proposition \ref{pro-AB}, $T_1$ has a sibling, let it be $T_2$. Let $_1$ and $R_2$ be the roots respectively.
We denote $shared(N)$ all shared attributes of node/relation $N$. We have two cases: 

$shared(R_1)=shared(R_2)$. In this case, if there is a break in the stem of $T_1$, then the parent of the root  $R_1$ of $T_1$ will be $R_2$ according to Theorem \ref{thm-break-general}. If there is no break in neither stems then we consider $B_S$ to be the set of attributes in the stems of $T_1$ and $T_2$ after the removal of the shared attributes.  $B_S$ has the properties of  Definition \ref{dfn-leads-to}. 

$shared(R_1)\neq shared(R_2)$. Then, the lowest common ancestor of $R_1$ and $R_2$ contains all the shared attributes and since neither $R_1$ nor $R_2$ contains all of them, there is another maximal subtree, say $T_3$ (with root $R_3$) which does. Thus, if there is a break in the stem of $T_1$, then the parent of the root  $R_1$ of $T_1$ wil be $R_3$, according to Theorem \ref{thm-break-general}.  If there is no break in neither stems then we consider $B_S$ to be the set of attributes in the stems of $T_1$ and $T_2$ after the removal of the shared attributes.  $B_S$ has the properties of  Definition \ref{dfn-leads-to}. 
%
%
\end{proof}

\section{Proofs for Section  \ref{sec-min}}
\label{app-newsec}

\begin{theorem}
\label{thm-0}
If for parse tree $T$ there is an  general break then the algorithm produces in one change a parse tree $T'$ with strictly fewer maximal subtrees in the subjoin. 
\end{theorem}
\begin{proof}
The
reverse path transformation from Section \ref{subsec-reverse} guarantees that the  algorithm (i.e., delete arc $e_1$ and add $e'_1$) and the second operation  creates a parse tree. Observe that deleting $e_1$ leaves two disjoint parse trees (disjoint wrto the attributes) up to the set of attributes that appears along all nodes of the path that defines the general break. Thus $e_2$ is an arc that creates a parse tree since both its ends contain all attributes appearing along the path that defines the general break.
\end{proof}

%
%
%

%
%
%
We prove now that, if there is no general break, then, we cannot find a parse tree with smaller number of maximal subtrees than the current parse tree $T$. 

\begin{theorem}
\label{thm-1}
If there is a parse tree with a general break that is produced from a parse tree without a general break in a certain number of changes, then there is a (final) parse tree without a general break that produces a parse with a general break in one change.
\end{theorem}

\begin{proof}
Evident
\end{proof}

\begin{theorem}
\label{thm-2}
If there is no general break, then we cannot find with one change a parse tree with a general break and simultaneously retaining the same or smaller number of maximal subtrees in the subjoin.
\end{theorem}

\begin{proof}

Suppose we have parse tree $T$ without a general break and after one change we obtain parse tree $T'$ with a general break with respect to maximal subtrees $T_1$ and $T_2$. 
Suppose edge $e$ is in $T$ but not in $T'$ and edge $e'$ is in $T'$ but not in $T$ and it is the edge that introduces the general break. Now consider the set $S$ of attributes that are shared between the two edges/nodes of $e'$. 
We have two cases as follows:

In the first case we assume that $e$ joins two nodes in the subjoin. After deleting it, we have more maximal subtrees in $T'$.

In the second case, we assume that one of the ends of $e$ is a non-subjoin node. Then we will  showthat $e$ defines a general break point  with respect to maximal subtrees $T_1$ and $T_2$. Since $e'$ is not in $T$, all the attributes that are shared by the endpoints of $e'$ should appear in the same connected component of $T$, hence, they should appear in the endpoints of $e$ too, since the edges $e$ and $e'$ are on a cycle (if both are included), and hence this is the only way to satisfy the condition of the definition of acyclicity.
%
%
%
\end{proof}

\begin{theorem}
\label{thm-n}
If there is no general break then we cannot find with one change (i.e., add arc, delete arc) a parse tree with fewer maximal subtrees. 
\end{theorem}

\begin{proof}
Suppose there is no general break and we can find with one change a parse tree  $T_n$ with fewer maximal subtrees. Then, this means that there are  two nodes $N_a$ and $N_b$, each from different maximal subtree in $T$ (say subtrees $T_1$ and $T_2$) such that they have a parent/child relationship in $T_n$, let us call $e$ the arc that denotes this parent/child relationship. (Arc $e$, thus, appears in $T_n$ but not in $T$.) If we add $e$ in $T$, we will create a cycle, thus, in $T_n$ some arcs of this cycle are not present. Suppose $e'$ is such an arc that is not present in $T_n$. For the property that each attribute must appear in a connected part of any parse tree to hold in $T_n$, $e'$ should connect two nodes $N_1$ and $N_2$ such that their shared attributes (i.e., $N_1\cap N_2$)
appear along the path $p$ from $N_1$ to $N_2$, where $p$ is formed by the following paths  in $T$: a path  from $N_a$ to $N_1$ and a path from $N_b$ to $N_2$, and arc $e$. This can only happen if there is a general break  
$(N_1,N_2)$ with respec to  $T_1$ and $T_2$, i.e., when $N_1\cap N_2 - \{SharedAttributes\}$ is empty (where $\{SharedAttributes\}$ is the set of all attributes shared by $N_1$ and $N_2$). Because otherwise, the set  $N_1\cap N_2$ contains attributes that are not shared between any two of the maximal subtrees, and, hence, an arc should exist  in $T_n$ that makes connected all the nodes that a certain such attribute appear. Such an arc however will create  cycle with the arc $e$ because it will create another path from $N_1$ to  $N_2$ in $T_n$. 
%
%
 \end{proof}

Putting  it all together, Theorem~\ref{thm-0} says that , if there is a general break, then we can find a parse tree with strictly fewer maximal subtrees in the subjoin. Theorem~\ref{thm-1} with  Theorem~\ref{thm-2} imply that if there is no break then we cannot find a break after any number of changes. And Theorem~\ref{thm-n} concludes by 
saying that the only way to find a parse tree with strictly fewer maximal subtrees is by using a general break.

\end{document}